\documentclass[conference]{IEEEtran}
\IEEEoverridecommandlockouts
\usepackage{booktabs}
\usepackage{subfiles}
\usepackage{colortbl}
\usepackage{subfigure}
\usepackage{graphicx}
\usepackage{caption}
\usepackage{amsthm}
\newtheorem{theorem}{Theorem}[section]
\newtheorem{definition}[theorem]{Definition}
\newtheorem{lemma}[theorem]{Lemma}

\newtheorem{example}[theorem]{Example}
\newtheorem{remark}[theorem]{Remark}
\def\*#1{\mathbf{#1}}
\def\~#1{\boldsymbol{#1}}

\usepackage{cite}
\usepackage{amsmath,amssymb,amsfonts}
\usepackage{algorithmic}
\usepackage{graphicx}
\usepackage{textcomp}
\usepackage{xcolor}
\def\BibTeX{{\rm B\kern-.05em{\sc i\kern-.025em b}\kern-.08em
    T\kern-.1667em\lower.7ex\hbox{E}\kern-.125emX}}
\begin{document}

\title{Attribution Methods in Asset Pricing: Do They Account for Risk?
}

\author{\IEEEauthorblockN{Dangxing Chen
\IEEEauthorrefmark{1} 
\thanks{$^*$ Corresponding author.}
}
\IEEEauthorblockA{\textit{Zu Chongzhi Center for Mathematics and Computational Sciences} \\
\textit{Duke Kunshan University}\\
Kunshan, China \\
dangxing.chen@dukekunshan.edu.cn}
\and
\IEEEauthorblockN{Yuan Gao}
Charlotte, U.S. \\
gaoyuanmath@gmail.com
}

\maketitle

\begin{abstract}
Over the past few decades, machine learning models have been extremely successful. As a result of axiomatic attribution methods, feature contributions have been explained more clearly and rigorously. 
  There are, however, few studies that have examined domain knowledge in conjunction with the axioms. 
  In this study, we examine asset pricing in finance, a field closely related to risk management. Consequently, when applying machine learning models, we must ensure that the attribution methods reflect the underlying risks accurately. In this work, we present and study several axioms derived from asset pricing domain knowledge. It is shown that while Shapley value and Integrated Gradients preserve most axioms, neither can satisfy all axioms. Using extensive analytical and empirical examples, we demonstrate how attribution methods can reflect risks and when they should not be used. 
\end{abstract}


\section{Introduction}
Recent decades have seen a tremendous amount of success with machine learning (ML). The use of ML increases the accuracy of traditional statistical models, but often at the expense of transparency and explainability. In industries with high stakes, such as the financial industry, explainability is essential. In the model risk management handbook\footnote{https://www.occ.treas.gov/publications-and-resources/publications/comptrollers-handbook/files/model-risk-management/index-model-risk-management.html} recently released by the Office of the Comptroller of the Currency, it is stressed that ML models must be explained in a clear and concise manner when applied. This has led to extensive research on explainable machine learning \cite{lundberg2017unified,ribeiro2016should,horel2020significance,sundararajan2017axiomatic}. 

Our work is focused on the attribution problem, i.e., how the ability to ascribe a value to a base feature can be interpreted as its role in predicting. A leading approach to attribution is based on the Shapley value by \cite{shapley1953value}, other popular methods include Integrated Gradients (IG) by \cite{sundararajan2017axiomatic}. An elegant aspect of these approaches is the use of axioms as a guide. By utilizing a few axioms, attribution methods can be uniquely determined. 

In spite of the success of the axiomatic approach, previous studies focused solely on the axioms of general models with no prior domain expertise. However, domain knowledge has been extensively studied in science throughout history. In recent studies, it has been demonstrated that when knowledge is incorporated, ML models become more reasonable and could potentially achieve higher accuracy \cite{gupta2020multidimensional,chen2023address}. In this regard, we should not neglect the domain knowledge implied in ML models when explaining them. 

The knowledge required for different domains varies. This study focuses on finance, a high-risk industry. Specifically, we focus on the issue of asset pricing. In finance, risk management is crucial to making appropriate decisions and avoiding significant financial losses. Inadequate risk management could lead to catastrophic consequences. For example, we just witnessed the collapse of Silicon Valley Bank (SVB), which is considered to be the second-largest bank failure in the United States. Upon the failure of SVB, a ripple effect was felt throughout the financial system, causing the stock market to go into a panic. The improper management of risk is one of the major causes of its failure, as outlined in \cite{barr2023review}. In particular, the interest rate risk has been neglected by SVB. In the short version, SVB's portfolio was extremely sensitive to interest rates, and the interest rate continued to increase, resulting in tremendous losses. A careful risk management strategy could have prevented the tragedy of SVB. In the wake of the failure of SVB, we have learned a painful lesson that risk management is an extremely important component of finance, and we should be aware of various risks on a regular basis. As a result, we ask the following question: \textbf{When applying attribution methods to financial models, can attribution methods accurately capture the risks implicit in the model?}

We explore attribution problems for asset pricing in order to answer the above question. We propose a number of axioms to reflect risks. These axioms are carefully analyzed in relation to Shapley value and IG. Fortunately, both attribution methods are capable of preserving most of the axioms. IG, however, cannot maintain demand monotonicity, as if the model is monotonic in a certain feature, then the attributions for that feature increase as its value increases; Shapley value may involve calculations outside of the training domain, which could present difficulties such as times of financial crisis. Using extensive analytical and empirical examples, we demonstrate how attribution methods can reflect risks and when their application might be detrimental. 

Our work has made the following contributions:
\begin{itemize}
    \item Several axioms are proposed based on different aspects for reflecting risk in asset pricing. 
    \item Our examples demonstrate the motivations and significance of these axioms.
    \item We provide a thorough analysis of Shapley value and IG, outlining their advantages and disadvantages. 
\end{itemize}



\section{Preliminaries}
\label{prerequisites}
For problem setup, assume we have $n$ features. We denote a class of functions $f: \mathbb{R}^n \rightarrow \mathbb{R}$ by $\mathcal{F}$.  We assume $\mathcal{F}$ to be the set of real analytic functions. For simplicity, we omit the discussion of nonanalytic functions that can be approximated to arbitrary precision by analytic functions.

\subsection{Attribution Methods}
Following \cite{lundstrom2022rigorous}, we call the point of interest $\overline{\*x}$ to explain as an explicand and $\*x'$ a baseline. Without loss of generality (WLOG), we assume $\overline{\*x} \geq \*x'$. The Baseline Attribution Method that interprets features' importance is defined below.

\begin{definition}[Baseline Attribution Method (BAM)]
    Given $\overline{\*x}, \*x' \in \mathbb{R}^n$, $f \in \mathcal{F}$, a baseline attribution method is any function of the form $\*A: \mathbb{R}^n \times \mathbb{R}^n \times \mathcal{F} \rightarrow \mathbb{R}^n$. 
\end{definition}

The Shapley value \cite{shapley1951notes} takes as input a set function $v:2^N \rightarrow \mathbb{R}$, which produces attributions $s_i$ for each player $i \in N$ that add up to $v(N)$, where $N = \{1, \dots, n\}$.

\begin{definition}[Shapley value]
    The Shapley value of a player $i$ is given by:
\begin{align}
    s_i = \sum_{S \subseteq N \backslash i} \frac{|S|! (|N|-|S|-1)!}{|N|!} (v(S \cup i) - v(S)).
\end{align}
\end{definition}
We focus on the Baseline Shapley (BShap) \cite{sundararajan2020many}, in which
\begin{align}
    v(S) = f(\overline{\*x}_S; \*x'_{N \backslash S}).
\end{align}
That is, baseline values replace the feature's absence. We denote BShap attribution by $\text{BS}_i(\overline{\*x},\*x',f)$ and $\text{BS}_i$ sometimes. For example, suppose $f(x_1,x_2) = x_1+x_2$, $\overline{\*x} = (\overline{x}_1, \overline{x}_2)$, $\*x' = (0,0)$, and $S=\{1\}$, then we have $v(S) = f(\overline{x}_1,0)$. We focus on the BShap since it has better theoretical properties than SHAP in terms of preserving axioms, as discussed in \cite{sundararajan2020many}. As a result, we are more confident about using it for sectors with high stakes. 

Another popular method is the Integrated Gradients \cite{sundararajan2017axiomatic}.
\begin{definition}[Integrated Gradients (IG)]
    Given $\overline{\*x}, \*x' \in \mathbb{R}^n$ and $f \in \mathcal{F}$, the Integrated Gradients attribution of the $i$-th component of $\overline{\*x}$ is defined as 
    \begin{align}
        \text{IG}_i(\overline{\*x},\*x',f) = (\overline{x}_i-x_i') \int_0^1 \frac{\partial f}{\partial x_i} \left( \*x' + t(\overline{\*x}-\*x') \right) \ dt.
    \end{align}
\end{definition}
For simplicity, we often use $\text{IG}_i$ for $\text{IG}_i(\overline{\*x},\*x',f)$.

\section{Asset Pricing}\label{sec:assets}

Asset pricing examines how financial assets are valued, such as stocks, financial derivatives, and bonds. Since the Nobel Prize work by \cite{black1973pricing,merton1973theory}, stochastic differential equations with risk-neutral pricing have been used as primary tools for pricing \cite{ahn1999parametric,heston1993closed,vasicek1977equilibrium}. In recent years, ML models have been increasingly used as an approximation of pricing formulas \cite{bali2023option,horvath2021deep}. These studies have demonstrated that ML models offer advantages over classical stochastic methods in terms of more flexible approximation formulas, computational efficiency, and the ability to be data-driven with fewer assumptions. 

Explainability is crucial in sectors such as finance where stakes are high. The application of ML models must be accompanied by an explanation. Existing explainable ML tools have, however, focused primarily on general problems without domain knowledge. As a result, explanations may be insufficient. In the finance community, a considerable amount of attention is paid to mathematical derivatives of pricing formulas, both first-order and higher-order. Therefore, asset pricing considers mathematical derivatives to be domain knowledge. Using mathematical derivatives, researchers and practitioners have been able to gain a better understanding of financial models. For this reason, we intend to incorporate mathematical derivatives into our explanation of ML models. More specifically, \textbf{BAMs must reflect risks associated with sensitive features. }


We provide two examples with discussions on their mathematical derivatives and corresponding models under simplified assumptions, which will be used in the remainder of the manuscript. By using these examples, we demonstrate how researchers and practitioners interpret information derived from mathematical derivatives.

\subsection{Coupon Bonds} \label{sec:bond}
A coupon is an interest payment received by a bondholder from the date of issuance until the maturity date of the bond. Zero coupon bonds are the simplest form of coupons. A zero coupon bond with a principal amount of $\$c$ and a maturity time $T$ will pay $\$c$ at $T$. 
\begin{example}\label{eg:bond}
Using continuous compounding, assume the interest rate is constant $r$, the present value ($t=0$) of a zero coupon bond is calculated as follows:
\begin{align}
    B(r,c,T) = ce^{-rT}.
\end{align}
\end{example}
The first derivative of a bond based on its interest rate, $\frac{\partial B}{\partial r}$, is known in finance as related to \textbf{duration}, introduced by \cite{macaulay1938some}. In general, the longer the duration, the more sensitive the bond price is to changes in interest rates. A second-order derivative, $\frac{\partial^2 B}{\partial r^2}$ is known in finance as related to \textbf{convexity} by \cite{dattatreya1991fixed}.  Convexity is usually preferred as it reduces sensitivity to interest rates.

\subsection{Option Pricing}\label{sec:option}

Option contracts convey to their owners, the holders, the right, but not the obligation, to purchase or sell a specified quantity of an underlying asset or instrument at a specified strike price on or before a prescribed date. European call options are a classic example. 
\begin{example}\label{eg:call_option}
    A call option is a contract between the buyer and the seller to exchange a security at a strike price $K$ at a maturity time $T$. At time $T$, if the stock price $S_T$ exceeds the strike price, then the option will be exercised and the payoff will be $S_T-K$; otherwise, the option will not be exercised and the payoff will be $0$. In summary, the value of the call option at time $T$ is equal to $C(S_T,K,T) = (S_T-K)^+$.  In option pricing, the key question is: What is the present value of an option at time $t=0$? 

    Based on a couple of assumptions, Black, Scholes, and Merton \cite{black1973pricing,merton1973theory} developed a pricing formula $C(S_0,K,T,\sigma,r)$, where $\sigma$ represents constant volatility and $r$ represents the constant risk-free interest rate. 
\end{example}

In option pricing, mathematical derivatives are referred to as \textbf{Greeks} \cite{natenberg1994option}. Financial institutions typically set (risk) limits for each of the Greeks that their traders are not permitted to exceed \cite{hull2016options}. Suppose we denote $V$ as the price of a general option, some common first-order Greeks include Delta, which is calculated as $\frac{\partial V}{\partial S}$ to measure the sensitivity of the asset price to the option price, and Vega, which is calculated as $\frac{\partial V}{\partial \sigma}$ to measure the sensitivity of the volatility to the option price. Greeks of higher order are also commonly considered. Examples include Gamma $\frac{\partial^2 V}{\partial S^2}$ and Vanna $\frac{\partial^2 V}{\partial S \partial \sigma}$. Based on the domain knowledge, practitioners often are familiar with the implications of Greeks in a wide range of situations, so it is crucial to incorporate this knowledge into explanations. In the case of a call option, a positive Delta $ \left( \frac{\partial V}{\partial S} \right)$ implies that an increase in the underlying asset price should result in an increase in the option price. As a result, if BAMs are applied, they must provide consistent explanations. More details about Greeks can be found in Appendix~\ref{sec:Greeks}.

\section{Axioms From Domain Knowledge}


\subsection{Axioms on First-order Effects}
Monotonicity can be used to reflect first-order effects in asset pricing. WLOG, we assume that all monotonic features are monotonically increasing throughout the paper. Suppose $\~{\alpha}$ is the set of all individual monotonic features and  $\neg \~{\alpha}$ its complement, then the input $\*x$ can be partitioned into $\*x = (\*x_{\~{\alpha}}, \*x_{\neg \~{\alpha}})$. Individual monotonicity is defined as follows. 
\begin{definition}[\textbf{Individual Monotonicity}] \label{def:indi_mono}
$f$ is individually monotonic with respect to $\*x_{\~{\alpha}}$ if $\forall \*x,\*x^* \text{ s.t. } \*x_{\~{\alpha}} \leq \*x^*_{\~{\alpha}}$
\begin{align} \label{eq:mono_con1}
 & f(\*x) = f(\*x_{\~{\alpha}}, \*x_{\neg \~{\alpha}}) \leq f(\*x^*_{\~{\alpha}}, \*x_{\neg \~{\alpha}}) = f(\*x^*).
\end{align}
where $\*x_{\~{\alpha}} \leq \*x_{\~{\alpha}}^*$ means $x_{\alpha_i} \leq x_{\alpha_i}^*, \forall{i}$.
\end{definition}

As discussed by \cite{chen2023can}, applying BAMs to individually monotonic features should result in positive attributions. 

\begin{definition}[\textbf{Average Individual Monotonicity (AIM) Axiom}]
    Suppose $f$ is individually monotonic with respect to $x_{\alpha}$, then we say a BAM preserves average individual monotonicity if $\forall \overline{\*x} \text{ s.t. } \overline{\*x} \geq \*x'$, we have
    \begin{align}
        A_{\alpha}(\overline{\*x},\*x',f) \geq 0.
    \end{align}
\end{definition}

\begin{example}
    In Example~\ref{eg:bond}, the interest rate should attribute negatively. 
\end{example}

We expect in some situations that individual monotonicity will have a greater impact in that, whenever a feature is increased, its attribution should be increased accordingly. \cite{friedman1999three} introduced the concept of demand individual monotonicity for this situation. 

\begin{definition}[\textbf{Demand Individual Monotonicity (DIM) Axiom}]\label{def:DIM}
    Consider two explicands $\overline{\*x} = (\overline{x}_1, \dots, \overline{x}_n)$ and $\*x^* = (\overline{x}_1, \dots, \overline{x}_{\alpha}+c, \dots, \overline{x}_n)$, where $c >0$. Suppose $f$ is individually monotonic with respect to $x_{\alpha}$. We say a BAM preserves demand individual monotonicity if $\forall \overline{\*x} \text{ s.t. } \overline{\*x} \geq \*x'$,
    \begin{align}
        A_{\alpha}(\*x^*,\*x',f) \geq A_{\alpha}(\overline{\*x},\*x',f).
    \end{align}
\end{definition}

\begin{example}
    In Example~\ref{eg:bond}, any additional increase in the interest rate should always result in the decrease of a zero-coupon bond, regardless of its principal. 
\end{example}

\subsection{Axioms on Second-order Main Effects}

Second-order derivatives may provide additional insights into the model's behavior. The diminishing marginal effect is one of the most common phenomena \cite{gupta2018diminishing,pya2015shape}. 

\begin{definition}[\textbf{Diminishing Marginal Effect (DME)}]
Suppose $\*x = (x_{\alpha},\*x_{\neg \alpha})$. We say $f$ has the \textbf{diminishing marginal effect} with respect to $x_{\alpha}$ if $\frac{\partial}{\partial x_{\alpha}} f(x_{\alpha},\*x_{\neg \alpha}) \geq 0$ and $\frac{\partial^2}{\partial x_{\alpha}^2} f(x_{\alpha},\*x_{\neg \alpha}) \leq 0$. Similarly, we say $f$ has the \textbf{reverse DME (RDME)} if $\frac{\partial}{\partial x_{\alpha}} f(x_{\alpha},\*x_{\neg \alpha}) \leq 0$ and $\frac{\partial^2}{\partial x_{\alpha}^2} f(x_{\alpha},\*x_{\neg \alpha}) \geq 0$. 

\end{definition}

 Basically, DME implies a slowing in the increment of the function. This motivates us to propose the following axiom.

\begin{definition}[\textbf{Diminishing Marginal Axiom}] \label{def:DME_axiom}
    Suppose $f$ has the diminishing marginal effect with respect to $x_{\alpha}$, then we say a BAM preserves DME if for $\overline{x}_{\alpha} \geq x^*_{\alpha} > x_{\alpha}'$, we have
    \begin{align}
        \frac{A_{\alpha} ((\overline{x}_{\alpha},\overline{\*x}_{\neg \alpha}),\*x',f)}{\overline{x}_{\alpha}-x_{\alpha}'}  \leq \frac{A_{\alpha}((x^*_{\alpha},\overline{\*x}_{\neg \alpha}),\*x',f)}{x^*_{\alpha}-x_{\alpha}'}.
    \end{align}
    A BAM preserves reverse RDME if
    \begin{align}
        \frac{A_{\alpha} ((\overline{x}_{\alpha},\overline{\*x}_{\neg \alpha}),\*x',f)}{\overline{x}_{\alpha}-x_{\alpha}'}  \geq \frac{A_{\alpha}((x^*_{\alpha},\overline{\*x}_{\neg \alpha}),\*x',f)}{x^*_{\alpha}-x_{\alpha}'}.
    \end{align}
\end{definition}

\begin{example}
    In Example~\ref{eg:bond}, the bond has the RDME with respect to the interest rate. The bond benefits from such complexity, as it prevents huge losses when interest rates rise significantly. RDME axiom can, therefore, be used to reflect convexity in bonds. 
\end{example}

Similarly, we define increasing marginal effects. 

\begin{definition}[\textbf{Increasing Marginal Effect (IME)}]
Suppose $\*x = (x_{\alpha},\*x_{\neg \alpha})$. We say $f$ has the increasing marginal effect with respect to $x_{\alpha}$ if $\frac{\partial}{\partial x_{\alpha}} f(x_{\alpha},\*x_{\neg \alpha}) \geq 0$ and $\frac{\partial^2}{\partial x_{\alpha}^2} f(x_{\alpha},\*x_{\neg \alpha}) \geq 0$. 

\end{definition}

\begin{definition}[\textbf{Increasing Marginal Axiom}] \label{def:IME_axiom}
    Suppose $f$ has the increasing marginal effect with respect to $x_{\alpha}$, then we say a BAM preserves IME if for $\overline{x}_{\alpha} > x^*_{\alpha} > x_{\alpha}'$, we have
    \begin{align}
        \frac{A_{\alpha} ((\overline{x}_{\alpha},\overline{\*x}_{\neg \alpha}),\*x',f)}{\overline{x}_{\alpha}-x_{\alpha}'}  \geq \frac{A_{\alpha}((x^*_{\alpha},\overline{\*x}_{\neg \alpha}),\*x',f)}{x^*_{\alpha}-x_{\alpha}'}.
    \end{align}
\end{definition}

\subsection{Axioms on Comparing Assets}

The investor may wish to compare different assets before making a decision. Bonds as well as options discussed in Section~\ref{sec:assets} for example, may be considered together. As a result, different assets are involved, and their features may differ. Nevertheless, they share many common features, such as the market interest rate and volatility. If investors are concerned about the potential increase in interest rates, they may find it useful to compare the sensitivity of these assets. Therefore, we should be able to determine from BAMs if a particular product is always more sensitive to interest rates than another. This would allow investors to have a clear understanding of the risks associated with different assets. 

Consider $\*x = (x_{\alpha},\*x_{\~{\beta}},\*x_{\neg})$ and $\*y = (y_{\alpha},\*y_{\~{\beta}},\*y_{\neg})$. That is, $\*x$ and $\*y$ have the same features of $\alpha$ and $\~{\beta}$, but not necessarily the others, and we are mostly interested in the impact on $x_{\alpha}$. Note $\*x_{\neg}$ and $\*y_{\neg}$ might have different dimensions.

\begin{definition}[\textbf{First-order Monotonic Dominance (FMD)}]
    Suppose we have two functions $f(\*x)$ and $g(\*y)$. We say $f$ dominates $g$ with respect to $x_{\alpha}$ for the first-order if $\forall x_{\alpha} = y_{\alpha}, \forall \*x_{\~{\beta}} = \*y_{\~{\beta}}, \forall \*x_{\neg}, \*y_{\neg}$,
    \begin{align}
        & \frac{\partial}{\partial x_{\alpha}} f(x_{\alpha},\*x_{\~{\beta}},\*x_{\neg}) \geq \frac{\partial}{\partial y_{\alpha}} g(y_{\alpha},\*y_{\~{\beta}},\*y_{\neg}). 
    \end{align}
\end{definition}

\begin{definition}[\textbf{First-order Monotonic Dominance Axiom}] \label{def:FMD_axiom}
    Suppose $f$ dominates $g$ with respect to $x_{\alpha}$, then we say a BAM preserves monotonic dominance for the first-order if for two explicands $\overline{\*x}, \overline{\*y}$ such that $\overline{\*x}_{\~{\beta}} = \overline{\*y}_{\~{\beta}}, \overline{\*x}_{\~{\beta}}' = \overline{\*y}_{\~{\beta}}', \overline{x}_{\alpha} = \overline{y}_{\alpha}, \overline{x}_{\alpha}' = \overline{y}_{\alpha}'$, we have
    \begin{align}
        A_{\alpha} (\overline{\*x},\*x',f) \geq A_{\alpha} (\overline{\*y},\*y',g). 
    \end{align}
\end{definition}

Interestingly, FMD can be preserved if other axioms are maintained. This requires the introduction of a new axiom.

\begin{definition}[\textbf{Generalized Dummy (GD)}]
    We say a BAM preserves dummy if $\forall f \in \mathcal{F}$, if $f(\*x) = f(\*x^*)$, where $(\*x_*)_j = \*x_j$ except for $i$ for all $\*x, \*x^*$, then $A_i(\overline{\*x},\*x',f) = 0$. Furthermore, let $\*h(\*x) = (x_1, \dots, x_{i-1},x_{i+1},x_n)$ and let $g$ be a reduction of $f$ omitting dummy features. With loss of generality, let $g(x_1, \dots, x_{i-1},x_{i+1}, \dots, x_n) = f(x_1, \dots, x_{i-1}, 0, x_{i+1}, \dots, x_n)$, then we say a BAM preserves generalized dummy if $A_j(\*h(\overline{\*x}),\*h(\*x'),g) = A_j(\overline{\*x},\*x',f)$ for $j \neq i$.
\end{definition}

\begin{theorem}\label{thm:thm_FMD}
    If a BAM preserves linearity, generalized dummy, and AIM, then it preserves FMD. 
\end{theorem}

\begin{remark}[\textbf{Hedge}]
    The FMD imposes implications for hedging. If there are two assets A and B, and A is more sensitive to interest rates than B, then a BAM that preserves FMD could reflect this. Furthermore, if we long A and short B, then the BAM would be able to indicate that this strategy is still associated with a positive risk of interest rate. 
\end{remark}

\subsection{Axioms on the Domain} \label{sec:domain}

Last but not least, we present an axiom unrelated to mathematical derivatives. In finance, certain characteristics, such as stock prices and volatility, have a random nature. Because of this, even though their domains may be large, many events may only occur with a low probability. ML models cannot accurately predict such events if they have not been observed. As a rough guide, we can divide the domain into training and out-of-training domains based on the realization of the data. Furthermore, features are often correlated with each other. As a result, features are not necessarily distributed equally among themselves. When applying BAMs, we must ensure that functions are not evaluated outside the training domain. 

Consider Example~\ref{eg:call_option}, for which we are interested in the option price based on stock prices, volatility, interest rates, strike prices, and maturity dates. Although the current market liquidity allows us to collect large amounts of data with different strikes and maturity dates, stock prices and volatility can only be observed over time. Thus, we may only have limited data regarding stock prices and volatility, which are highly correlated. Therefore, it is crucial to identify their training domain. We illustrate our point with the example below.

We study S$\&$P 500 data during the 2008 financial crisis. More details can be found in Appendix~\ref{sec:VIX} and \ref{sec:NN}. Imagine that we have a model $f$ for option prices and we would like to determine the causes of the significant changes in prices before and during the market crash. VIX is used as the approximation for the volatility, with details in Appendix~\ref{sec:VIX}. As an example, on August 1, 2008, we consider the baseline point $\*x'$ with the stock price was about 1270 and the volatility was about 0.23, suggesting the market was in a normal state; on November 21, 2008, we consider the explicand $\overline{\*x}$ with the volatility hit 0.81 and the stock price plummeted to 756, suggesting the market is panicking. If we would like to investigate the impact of stock prices and volatility, we will need to identify the training domain. It is important to note that we did not train the model in the entire rectangle determined by $(S',\sigma'),(S',\sigma),(S,\sigma'),(S,\sigma)$, as shown in Figure~\ref{fig:S_vs_V}. Particularly, we do not have data on high stock prices with high volatility or low stock prices with low volatility. It has been consistently observed empirically that changes in an asset's volatility are negatively correlated with its return. In finance, this phenomenon is known as the leverage effect \cite{ait2013leverage,black1976studies}. Therefore, BAMs should avoid using function values in these out-of-training areas. Motivated by this, we propose the following axiom regarding the geometry of a domain. 

\begin{definition}[\textbf{Convex Geometry Axiom}]
    Suppose $\mathcal{X}$ is the training domain of $f(\*x)$. We say a BAM preserves convex geometry (CG) if $\mathcal{X}$ is convex, $\overline{\*x}, \*x' \in \mathcal{X}$, then the calculation of $\*A(\overline{\*x},\*x',f)$ only requires $f(\mathcal{X})$.
\end{definition}



\begin{figure}[h]
  \centering
  \includegraphics[scale=0.3]{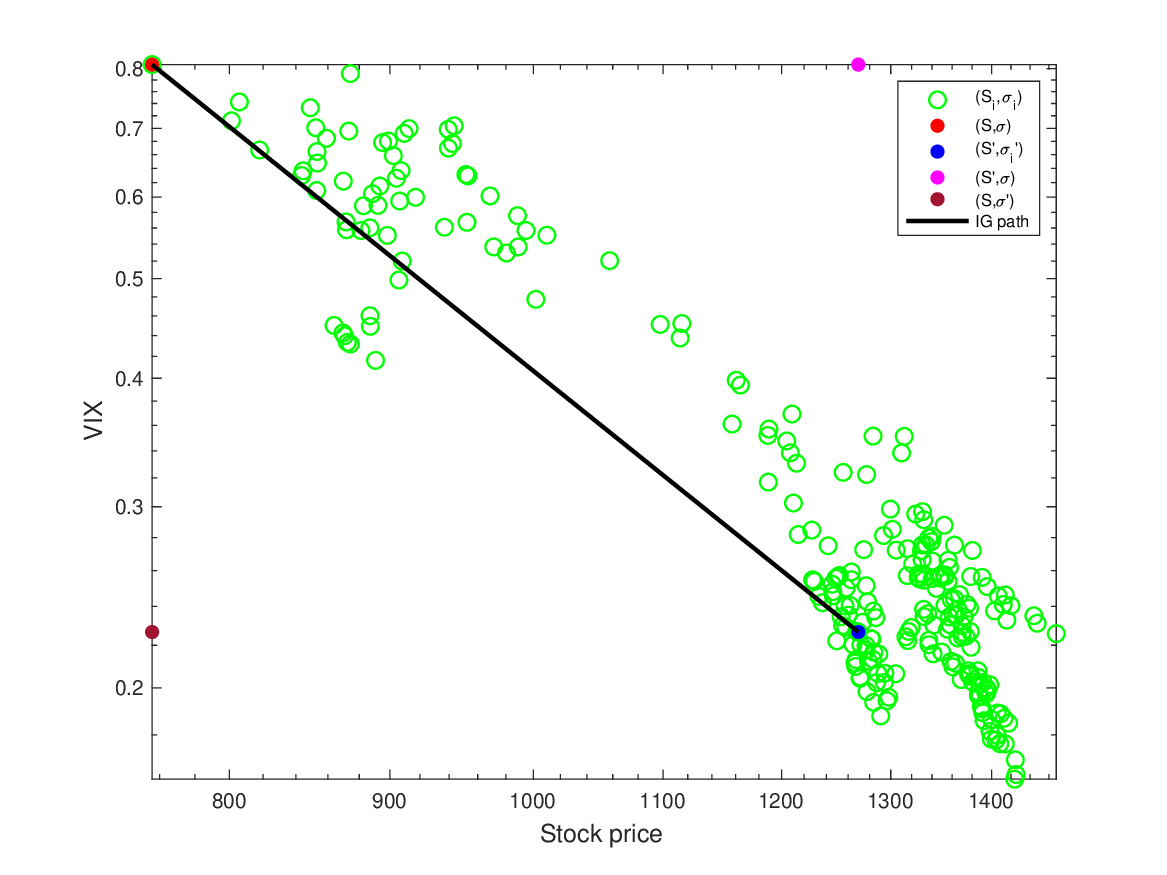}
  \captionof{figure}{Stock price vs VIX}
  \label{fig:S_vs_V}
\end{figure}


\section{Results for BShap and IG}\label{sec:THMs}

The results of BShap and IG for axioms are summarized and compared in detail. Proofs are left in Appendix~\ref{sec:proof}.  

\begin{theorem}\label{thm:BShap_all}
    BShap preserves average individual monotonicity (AIM), demand individual monotonicity (DIM), (reverse) diminishing marginal effect ((R)DME), increasing marginal effect (IME), and first-order monotonic dominance (FMD).
\end{theorem}

\begin{theorem}\label{thm:IG_all}
    IG preserves average individual monotonicity (AIM), (reverse) diminishing marginal effect ((R)DME), increasing marginal effect (IME), first-order monotonic dominance (FMD), and convex geometry (CG). 
\end{theorem}

As a result of the comparison, IG fails to preserve the DIM, which may be one of its greatest weaknesses. We provide a detailed example below. 

\begin{example}
    Consider Example~\ref{eg:bond}. Suppose we are interested in the explanation in terms of $r$ and the baseline is $\*x' = \*0$. By calculation, we have $
        \text{IG}_{\overline{r}} = \overline{c} \left( e^{-\overline{r}t} + \frac{e^{-\overline{r}t}}{\overline{r}t} - \frac{1}{\overline{r}t} \right).
    $
    For $\overline{r}t \sim 0$,  we have
    $
        IG_{\overline{r}} \sim - \frac{\overline{c}\overline{r}t}{2}.
    $
    As a result, DIM is preserved by IG. However, for $\overline{r}t \rightarrow \infty$, we have
    $
        \text{IG}_{\overline{r}} \sim -\frac{\overline{c}}{\overline{r}t},
    $
    whereas the DIM is violated. A similar event could occur for long-term bonds with a very high interest rate. As an example, inflation in the U.S. reached an extremely high level during the 1980s. The bond price decreases because of the rising interest rate but might not be reflected by IG. 
\end{example}

We examine when IG is capable of preserving DIM. 

\begin{theorem}\label{thm:IG_DIM}
    If $f$ has the IME with respect to $x_{\alpha}$, then IG preserves DIM. 
\end{theorem}

\begin{example}
    Consider the problem of option pricing in Section~\ref{sec:option}. As a reminder, Delta and Gamma are the first derivative and second derivative of option prices with respect to stock prices, respectively. According to the BSM, they are both positive for the call options. Thus, DIM can be preserved for the stock price by IG for call options. 
\end{example}         

In the case of BShap, its main weakness is the inability to preserve CG. 

\begin{example}
BShap doesn't preserve CG. Consider the example provided in Section~\ref{sec:domain}. When applying BShap to investigate the attributions of stock prices and volatility before and during the market crash, we are required to provide $f(S',v')$, $f(S',v)$, $f(S,v')$, and $f(S,v)$. $f(S',v)$ and $f(S,v')$ are, however, outside the training domain due to the leverage effect, as shown in Figure~\ref{fig:S_vs_V}. In contrast, the IG path appears to be a reasonable choice since there is some data closing to it. Therefore, BShap may not be the best option in this situation. 
\end{example}

\section{Empirical Results}

\subsection{Option Pricing}

We present an empirical study of option pricing in 2008 in the U.S., which is considered the worst financial crisis of the 21st century. This period is used to illustrate the results under extreme circumstances, but other periods can also be analyzed similarly. We collect transaction data for European call and put options. Models are based on five features, namely underlying stock (index) prices $S$, risk-free interest rates $r$, time to expiration $\tau$, strike prices $K$, and volatility $\sigma$ such that $\*x = (S,r,\tau,K,\sigma)$. For the sake of simplicity, we concentrate on these features, since they are regarded as the most important factors when pricing options. However, it may be possible to incorporate more features, such as alternative data, in order to improve the model's performance. These five features are applied to neural networks. A detailed description of the data and model setup is provided in Appendix~\ref{sec:NN}.

\subsubsection{An Example of Explanations}

We use the baseline point of 
\begin{align}\label{eq:eg1}
    \*x' = \left[ \begin{matrix} 1433.8 & 4.26 & 0.59 & 1396 & 0.23 \end{matrix} \right],
\end{align}
which is the average statistics on the first date for call options. It should be noted that this is not a unique choice, we only use it as a demonstration. We wish to explain the explicand 
\begin{align}\label{eq:eg1_expli}
    \overline{\*x} = \left[ \begin{matrix} 1344.8 & 3.09 & 0.2 & 1150 & 0.27 \end{matrix} \right],
\end{align}
which is a call option on February 5. As of this date, the market remains relatively calm. IG and BShap results are shown in Figure~\ref{fig:eg1_IG} and \ref{fig:eg1_shap}. Overall, our explanations are qualitatively similar. The major attributions are based on the difference in stock prices and strike prices, as expected. As a result of the shorter expiration date, stock prices are less likely to change, which has a negative and significant impact. Interest rates have only a small impact. Interest rates are usually not attributed to stock prices significantly, as stock prices are much more volatile than interest rates. As volatility has not changed too much, it is expected that attributions will be small.



\begin{figure}[h]
    \centering
    \begin{minipage}{0.3\textwidth}
        \centering
        \includegraphics[width=\linewidth]{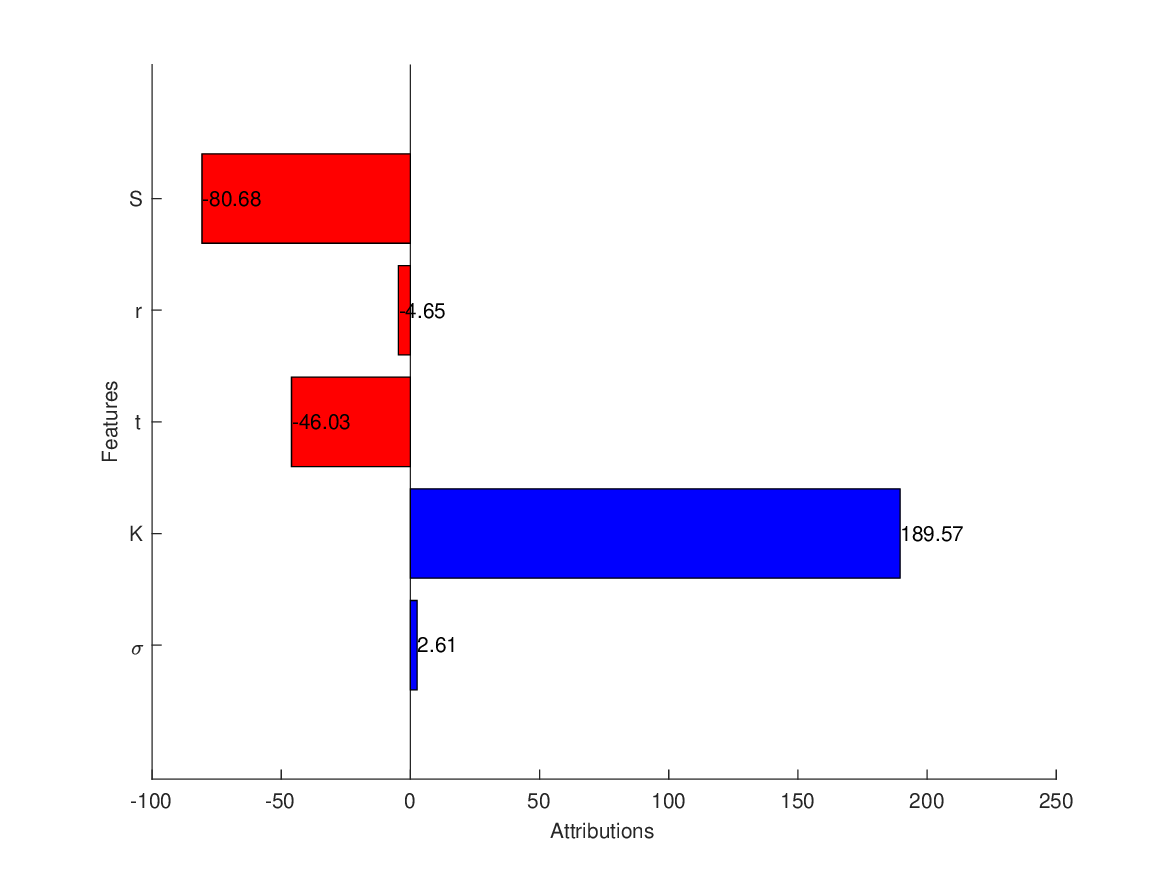}
        \caption{IG for Eq~\ref{eq:eg1}}
        \label{fig:eg1_IG}
    \end{minipage}\hfill
    \begin{minipage}{0.3\textwidth}
        \centering
        \includegraphics[width=\linewidth]{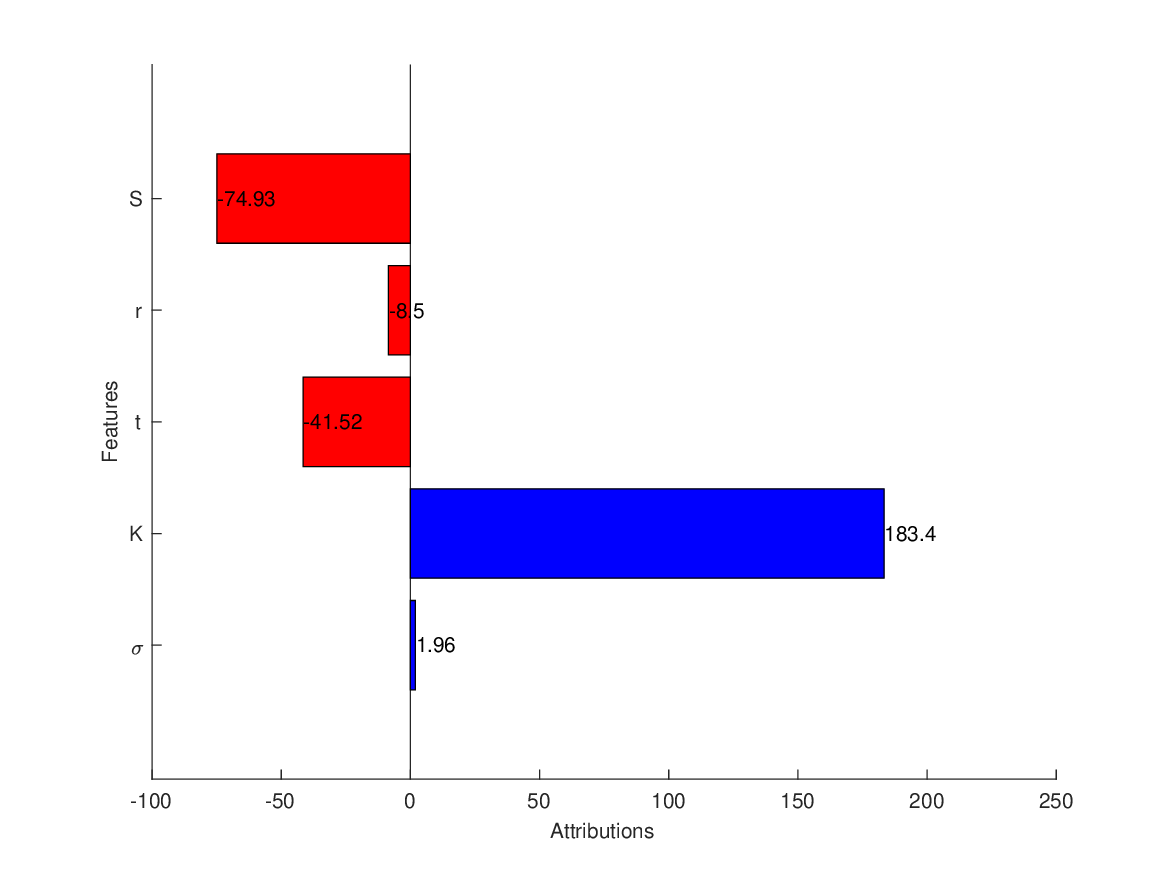}
        \caption{BShap for Eq~\ref{eq:eg1}}
        \label{fig:eg1_shap}
    \end{minipage}
\end{figure}


\subsubsection{Preservation/Violation of Risk Patterns}
Risk patterns of option pricing could also be observed. In this example, we vary the stock price from 1300 to 1500 and fix other parameters in \eqref{eq:eg1_expli} as well as the baseline \eqref{eq:eg1}. Based on the finance theory, when it comes to Delta $\frac{\partial V}{\partial S}$, it is positive for call options and negative for put options. Accordingly, we observe consistent explanations (Definition~\ref{def:IME_axiom}) in Figure~\ref{fig:eg1_term_call}. Gamma $\frac{\partial^2 V}{\partial S^2}$ are positive for both call and put options. Thus, we observe both convexity (Definition~\ref{def:DME_axiom}) for IG and BShap for a put option in Figure~\ref{fig:eg1_term_put}. The results validated the Theorem~\ref{thm:BShap_all}, \ref{thm:IG_all}, and demonstrated the potential of both IG and BShap in preserving risk patterns. It is also possible to observe other risk patterns. 




\begin{figure}[h]
    \centering
    \begin{minipage}{0.35\textwidth}
        \centering
        \includegraphics[width=1\textwidth]{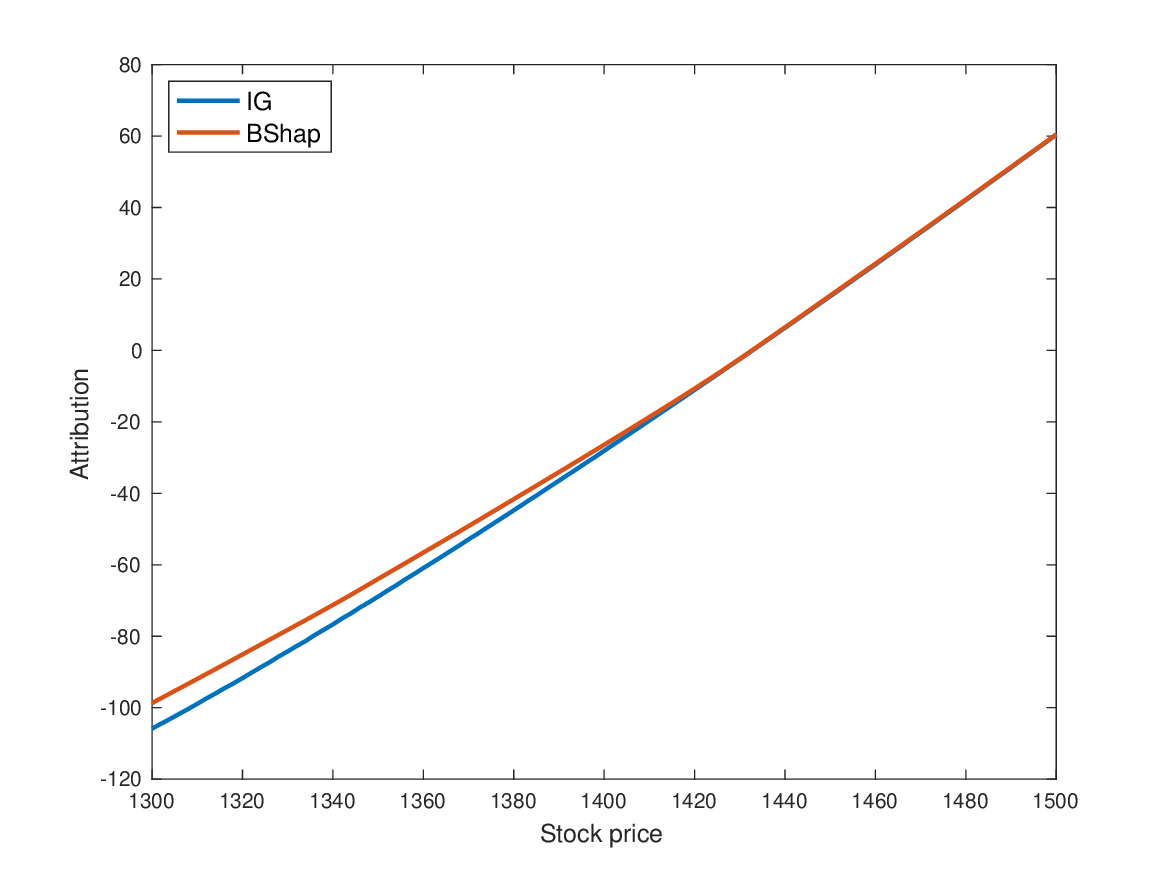} 
        \caption{Preservation of IME}
    \label{fig:eg1_term_call}
    \end{minipage}
    \begin{minipage}{0.35\textwidth}
        \centering
        \includegraphics[width=1\textwidth]{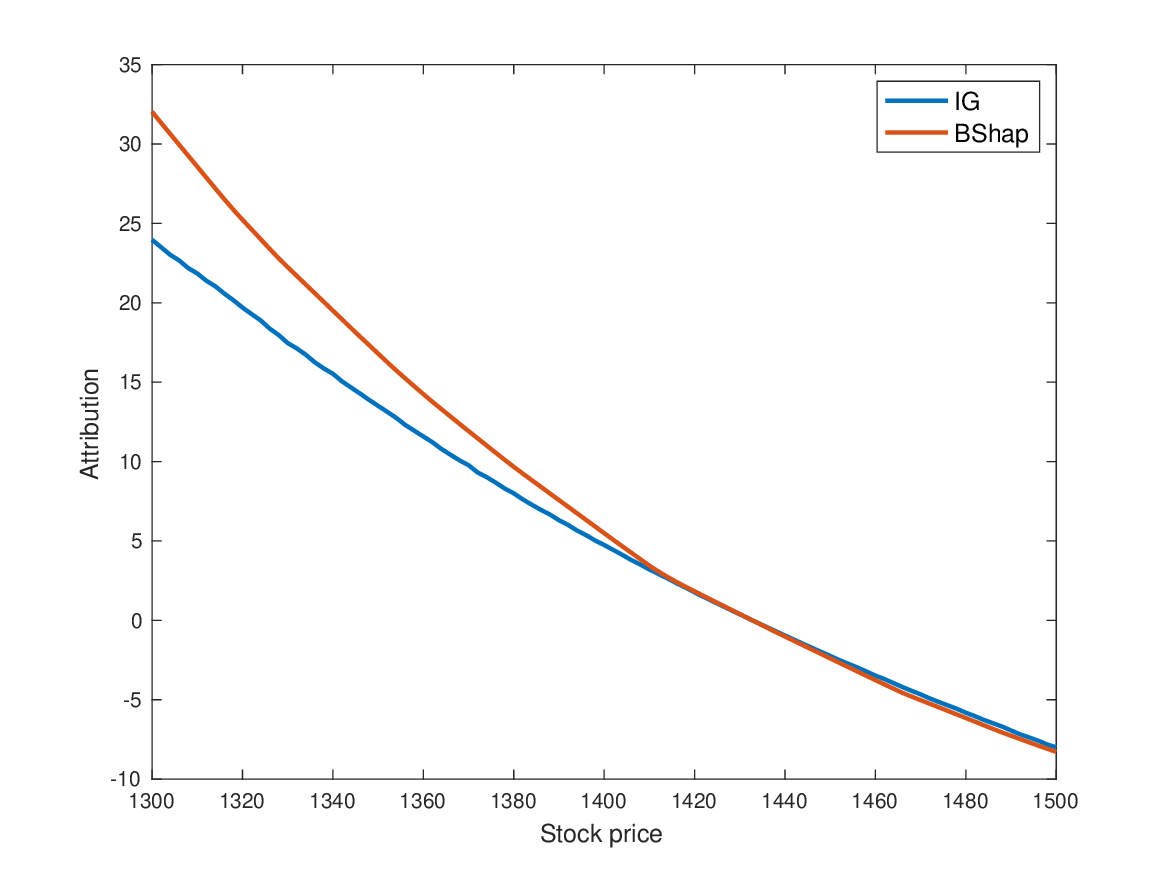}  \caption{Preservation of RDME}
        \label{fig:eg1_term_put}
    \end{minipage}
    \begin{minipage}{0.35\textwidth}
        \centering
        \includegraphics[width=1\textwidth]{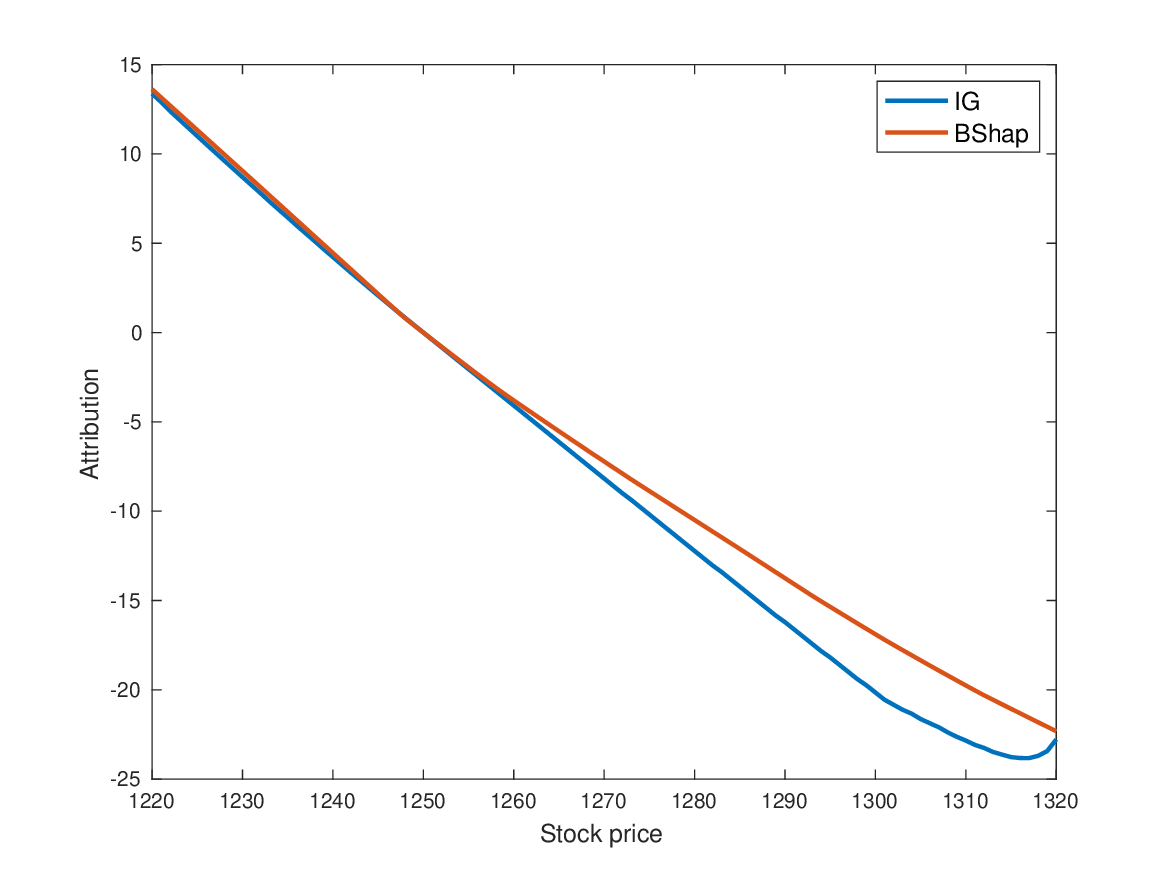} 
        \caption{Violation of DIM by IG}
        \label{fig:DIM_violation}
    \end{minipage}
\end{figure}

However, as discussed in Section~\ref{sec:THMs}, not all risk patterns are preserved by IG. As an example, we look at a put option with negative Delta $\frac{\partial V}{\partial S}$ and positive Gamma $\frac{\partial^2 V}{\partial S^2}$. Our analysis indicates that DIM (Definition~\ref{def:DIM}) is preserved for BShap, but not necessarily for IG. As an example, we consider the baseline point of 
\begin{align}
    \*x' = \left[ \begin{matrix} 
    1250 & 4.2 & 0.3 & 1240 & 0.25 \end{matrix} \right],
\end{align}
an explicand
\begin{align}
    \overline{\*x} = \left[ \begin{matrix} 
    1300 & 4.0 & 0.3 & 1300 & 0.4 \end{matrix} \right],
\end{align}
and we vary stock prices between 1220 and 1320. The risk pattern is plotted in Figure~\ref{fig:DIM_violation}. DIM for IG is violated here since the option exhibits large Gamma at the money ($S$ close to $K$). BShap is preferred over IG in such cases.

\section{Discussion and Future Work}



Our analysis indicates that BShap and IG are able to reflect risks in most situations when it comes to attribution problems. However, there are certain situations in which we should be more cautious. First, IG does not preserve DIM in general, which is why BShap is preferred when such a risk pattern is crucial to an application. Second, BShap does not preserve CG, which is why it should not be used when out-of-training samples are required for calculations, such as during periods of financial crises. 

It would be interesting, in light of the limitations, to explore other attribution methods that could preserve desired axioms in the future. It might be possible, for example, to further incorporate the time series property to avoid out-of-training samples when dealing with time series data. In a different direction, it would be interesting to examine risk attribution with applications such as portfolio analysis \cite{shalit2021shapley}.

\bibliographystyle{plain}
\bibliography{sample-base}

\appendix

\section{Option Pricing} \label{sec:option_pricing}

\subsection{Greeks}\label{sec:Greeks}
Here, we provide information regarding some commonly used Greeks. We denote $V$ as the price of a general option. 
\begin{itemize}
    \item Delta: $\frac{\partial V}{\partial S}$. Positive for a call and negative for a put.
    \item Vega: $\frac{\partial V}{\partial \sigma}$. Vega is always positive for both call and put options as larger volatility implies more possibilities of stock prices. 
    \item Rho: $\frac{\partial V}{\partial r}$. Rho is positive for a call option and negative for a put option. Rho is less sensitive than others since stock prices are much more volatile than interest rates. 
    \item Gamma: $\frac{\partial^2 V}{\partial S^2}$. Positive for both call and put options. 
    \item Vomma: $\frac{\partial^2 V}{\partial \sigma^2}$. Positive for both call and put options.
\end{itemize}

\subsection{VIX} \label{sec:VIX}

VIX is the ticker symbol for the Chicago Board Options Exchange's CBOE Volatility Index. It is used as a model-free measure of the stock market's expectation of volatility based on S$\&$P 500 index options. The VIX we use is the 30-day expected volatility of the S$\&$P 500 index, more precisely the square root of a 30-day expected realized variance of the index. It is calculated as a weighted average of out-of-the-money calls and put options on the S$\&$P 500,
\begin{align}
    \text{VIX} = \sqrt{\frac{2e^{r\tau}}{\tau} \left( \int_0^F \frac{P(K)}{K^2} \ dK+ \int_{F}^{\infty} \frac{C(K)}{K^2} \ dK \right) },
\end{align}
where $F$ is the 30-day forward price on the S$\&$P 500. More details about the calculation can be found in \cite{whaley1993derivatives}.

\subsection{Proofs} \label{sec:proof}

\begin{proof}[Proof of Theorem~\ref{thm:thm_FMD}]
    Consider $\*z$ that generalize $\*x$ and $\*y$ and new functions $\widetilde{f}(\*z)$ and $\widetilde{g}(\*z)$ whereas features that are not used as dummy features. Due to the generalized dummy axiom, we know feature attributions are the same for new functions. Let $h(\*z) = \widetilde{f}(\*z) - \widetilde{g}(\*z)$, due to the linearity and dummy axioms, we have
    \begin{align*}
        A_{\alpha}(\overline{\*z},\*z',h) &= A_{\alpha}(\overline{\*z},\*z',\widetilde{f}) - A_{\alpha}(\overline{\*z},\*z',\widetilde{g}) \\
        &= A_{\alpha}(\overline{\*x},\*x',f) - A_{\alpha}(\overline{\*y},\*y',g).
    \end{align*}
    Due to AIM, we have
$
        A_{\alpha}(\overline{\*z},\*z',h) \geq 0.
$
    Thus, we conclude. 
\end{proof}

\begin{lemma}\label{lem:BShap_AIM_DIM}
    BShap preserves AIM and DIM. 
\end{lemma}
\begin{proof}
    Proof in \cite{chen2023can}.
\end{proof}

\begin{lemma}\label{lem:BShap_DME}
    BShap preserves DME, RDME, and IME. 
\end{lemma}

\begin{proof}
    We prove the results for DME, the reverse case and IME are similar. Suppose $f$ has the DME with respect to $x_{\alpha}$, then $\frac{\partial}{\partial x_{\alpha}} f(x_{\alpha},\*x_{\neg \alpha})$ is monotonically decreasing with respect to $x_{\alpha}$. Suppose we have explicands $\overline{\*x} = (\overline{x}_{\alpha},\overline{\*x}_{\neg})$ and $\*x^* = (x_{\alpha}^*,\overline{\*x}_{\neg})$. WLOG, suppose $\*x'=\*0$ (BShap preserves affine transformation \cite{sundararajan2020many}) and $f(\*x') = 0$. By the mean value theorem, we have
    \begin{align*}
        \frac{\partial}{\partial x_{\alpha}} f(c,\overline{\*x}_{\neg \alpha}) &= \frac{f(\overline{x}_{\alpha},\overline{\*x}_{\neg \alpha}) - f(x^*_{\alpha},\overline{\*x}_{\neg \alpha})}{\overline{x}_{\alpha} - x_{\alpha}^*}, \\
        \frac{\partial}{\partial x_{\alpha}} f(d,\overline{\*x}_{\neg \alpha}) &= \frac{f(x_{\alpha}^*,\overline{\*x}_{\neg \alpha})-0}{x_{\alpha}^*-0},
    \end{align*}
    where $\overline{x}_{\alpha}>c>x_{\alpha}^*>d>0$. Since $\frac{\partial}{\alpha x_{\alpha}} f(x_{\alpha},\*x_{\neg})$ is monotonically decreasing, $\frac{\partial}{\partial x_{\alpha}} f(c,\*x_{\neg \alpha}) \leq \frac{\partial}{\partial x_{\alpha}} f(d,\*x_{\neg \alpha})$. Therefore,
    \begin{align*}
        \frac{f(\overline{x}_{\alpha},\overline{\*x}_{\neg \alpha}) - f(x^*_{\alpha},\overline{\*x}_{\neg \alpha})}{\overline{x}_{\alpha} - x_{\alpha}^*} \leq \frac{f(x_{\alpha}^*,\overline{\*x}_{\neg \alpha})}{x_{\alpha}^*}.
    \end{align*}
    Now in the calculation of Shapley value, we have
    \begin{align*}
        \frac{s_{\alpha}-s_{\alpha}^*}{\overline{x}_{\alpha}-x_{\alpha}^*} &= 
        \sum_{S \subseteq N \backslash \alpha} \frac{|S|! (|N|-|S|-1)!}{N!} \frac{v(S \cup \alpha) - v^*(S \cup \alpha)}{\overline{x}_{\alpha}-x_{\alpha}^*} \\
        &\leq \sum_{S \subseteq N \backslash \alpha} \frac{|S|! (|N|-|S|-1)!}{N!} \frac{f(x^*_{\alpha},\overline{x}_S;\*x'_{N \backslash (S \cup \alpha))}) }{x_{\alpha}^*} \\
        &= \sum_{S \subseteq N \backslash \alpha} \frac{|S|! (|N|-|S|-1)!}{N!} \frac{v^*(S \cup \alpha)}{x_{\alpha}^*} = \frac{s_{\alpha}^*}{x_{\alpha}^*}.
    \end{align*}
    This implies that
    \begin{align*}
        \frac{s_{\alpha}}{\overline{x}_{\alpha}} \leq \frac{s_{\alpha}^*}{x_{\alpha}^*}.
    \end{align*}
    Thus, we conclude. 
\end{proof}

\begin{lemma}\label{lem:BShap_FMD}
    BShap preserves FMD. 
\end{lemma}

\begin{proof}
    The proof is followed by the preservation of linearity and the generalized dummy of Shapley Values and Theorem~\ref{thm:thm_FMD}. For the generalized dummy, Shapley value collects the marginal contribution of all orders of players. Suppose now there is an additional dummy feature, the presence of dummy features can be removed without affecting the calculation. Therefore, the calculation of non-dummy features is the same as the game omitting the dummy feature.  
\end{proof}

\begin{proof}[Proof of Theorem~\ref{thm:BShap_all}]
    By Lemmas~\ref{lem:BShap_AIM_DIM}, \ref{lem:BShap_DME},  \ref{lem:BShap_FMD}.    
\end{proof}


\begin{lemma}\label{lem:IG_DME}
    IG preserves AIM, DME, RDME, and IME. 
\end{lemma}

\begin{proof}
    Proof for the AIM is in \cite{chen2023can}. We prove the result for DME, the reverse case and IME case are similar. WLOG, suppose $\*x' = \*0$ (IG preserves affine transformation \cite{lundstrom2023four}). We calculate that
    \begin{align*}
        & \frac{1}{\overline{x}_\alpha} \text{IG}_{\alpha}(\overline{\*x}) - 
        \frac{1}{x_{\alpha}^*} \text{IG}_{\alpha}(\*x^*) \\
        &= 
        \int_0^1 \frac{\partial}{\partial x_{\alpha}} f(t \overline{x}_{\alpha},t \overline{\*x}_{\neg \alpha}) - \frac{\partial}{\partial x_{\alpha}} f(tx_{\alpha}^*,t\overline{\*x}_{\neg \alpha}) \ dt.
    \end{align*}
    Since 
    \begin{align*}
        \frac{\partial}{\partial x_{\alpha}} f(t \overline{x}_{\alpha},t \overline{\*x}_{\neg \alpha}) - \frac{\partial}{\partial x_{\alpha}} f(tx_{\alpha}^*,t \overline{\*x}_{\neg \alpha}) \leq 0, \forall t \in [0,1]
    \end{align*}
    by mean-value theorem. We conclude.
\end{proof}

\begin{lemma}\label{lem:IG_FMD}
    IG preserves FMD. 
\end{lemma}

\begin{proof}
    The proof is followed by the preservation of linearity and the generalized dummy of IG and Theorem~\ref{thm:thm_FMD}. For generalized dummy, since $g(\*h(\*x)) = f(\*x)$, for $j \neq i$, by chain rule, we have
    \begin{align*}
        \frac{\partial g}{\partial h_j} \frac{\partial h_j}{\partial x_j} = \frac{\partial f}{\partial x_j} \Rightarrow \frac{\partial g}{\partial h_j} = \frac{\partial f}{\partial x_j}.
    \end{align*}
\end{proof}

\begin{lemma}\label{lem:IG_CG}
    IG preserves CG. 
\end{lemma}

\begin{proof}
    By Definition.
\end{proof}

\begin{proof}[Proof of Theorem~\ref{thm:IG_all}]
    By Lemma~\ref{lem:IG_DME}, \ref{lem:IG_FMD}, \ref{lem:IG_CG}.
\end{proof}

\begin{proof}[Proof of Theorem~\ref{thm:IG_DIM}]

    Since $x_{\alpha}^* \geq x_{\alpha}$ and $\frac{\partial f}{\partial x_{\alpha}} \geq 0$, we have
    \begin{align*}
        & \text{IG}_{\alpha}(\*x^*,\*x',f) - \text{IG}_{\alpha}(\overline{\*x},\*x',f) \geq \\
        & (\overline{x}_{\alpha}-x_{\alpha}') \int_0^1 \frac{\partial f}{\partial x_{\alpha}}(\*x'+t(\*x^*-\*x')) - \frac{\partial f}{\partial x_{\alpha}}(\*x'+t(\overline{\*x}-\*x')) \ dt.
    \end{align*}
    By mean-value theorem and $\frac{\partial^2}{\partial x_{\alpha}^2} f \geq 0$, we have
    \begin{align*}
        \frac{\partial f}{\partial x_{\alpha}}(\*x'+t(\*x^*-\*x')) \geq \frac{\partial f}{\partial x_{\alpha}}(\*x'+t(\overline{\*x}-\*x')), \ \ \forall t \in [0,1].
    \end{align*}
    Thus, we conclude. 
    
\end{proof}

\subsection{Data and Neural Networks Setup for Option Pricing} \label{sec:NN}

We collected call and option data from Wharton Research Data Services in 2008. There are 253 trading days with 123969 records of option data in total. The London Interbank Offered Rate (LIBOR) is used to represent risk-free interest rates. The LIBOR served as the benchmark interest rate at which major global banks lent to one another in the international interbank market for short-term loans. This key benchmark interest rate served as an indication of borrowing costs between banks throughout the world. Daily volatility is measured by the VIX index. At each date $i$, we collect option data as $V(S_{i,j},r_{i,j},T_{i,j}, K_{i,j}, \sigma_i)$. Here are more explanations. On each date, options data with different stock prices, strike prices, and maturity dates are collected. Interest rates fluctuate over time. In addition, interest yields vary for different maturity times as a result of risk premiums. During a given day, volatility is assumed to be constant since the calculation of VIX requires access to a large number of option prices and is not possible to measure instantly.

For neural networks, we use the architecture of $[32,16]$ with Relu activations and $l2$ regularization with $\lambda = 10^{-3}$. We solve the optimization problem using the conjugate gradient, and we stop iterating after 1000 steps. We randomly split data into $75\%$ training data and $25\%$ test data. The error is measured by the squared root of the mean squared error. Two different neural networks are used to train call and put options. Call and put options result in errors of 3.73 and 3.85, respectively. This is somewhat larger than other periods given the volatility of the market in 2008. The neural network used here is only intended for demonstration purposes. More advanced models may provide better results, for example in \cite{dugas2009incorporating}.


\end{document}